\documentclass[11pt,onecolumn,twoside]{IEEEtran} 
\usepackage{amssymb,url,amsmath,amsthm}
\usepackage{here}
\usepackage[pdftex]{graphicx}
\usepackage{bm}
\usepackage{caption}
\theoremstyle{theorem}
\newtheorem{theorem}{Theorem}
\newtheorem{proposition}{Proposition}
\newtheorem{corollary}{Corollary}
\newtheorem{lemma}{Lemma}

\theoremstyle{definition}
\newtheorem{definition}{Definition}
\newtheorem{example}{Example}

\def\dmu{\mathrm{d}\mu}
\def\P{\mathcal{P}}
\def\L{\mathcal{L}}
\def\B{\mathcal{B}}

\def\O{\mathcal{O}}
\def\M{\mathcal{M}}
\def\<{\langle}
\def\>{\rangle}
\newcommand{\eqdef}{:=}

\newcommand{\argmin}{\mathop{\rm arg~min}\limits}

\title{Minimization Problems on Strictly Convex Divergences}
\author{Tomohiro Nishiyama \\ Email: htam0ybboh@gmail.com}

\begin{document} 
\maketitle
\bibliographystyle{plain}
\begin{abstract}
The divergence minimization problem plays an important role in various fields. In this note, we focus on differentiable and strictly convex divergences. For some minimization problems, we show the minimizer conditions and the uniqueness of the minimizer without assuming a specific form of divergences. Furthermore, we show geometric properties related to the minimization problems.

\end{abstract}

\section{Introduction}
The divergences are quantities that measure discrepancy between probability measures.
For two probability measures $P$ and $Q$, divergences satisfy the following properties. \\

$D(P\|Q)\geq 0$ with equality if and only if $P=Q$.\\

In particular, the $f$-divergence \cite {sason2016f, csiszar1967information}, the Bregman divergence \cite {bregman1967relaxation} and the R{\'e}nyi divergence \cite{van2014renyi,renyi1961measures} are often used in various fields such as machine learning, image processing, statistical physics, finance and so on.

In order to find the probability measures closest (in the meaning of divergence) to the target probability measure subject to some constraints, it is necessary to solve divergence minimization problems and there are many works about them \cite {csiszar2003information, breuer2016measuring, banerjee2005clustering}.
Divergence minimization problems are also deeply related to the geometric properties of divergences such as the projection from a probability measure to a set.

The main purpose of this note is to study minimization problems and geometric properties of differentiable and strictly convex divergences in the first or the second argument \cite {nishiyama2019monotonically}.
For example, the squared Euclidean distance is differentiable and strictly convex. The most important result is that we can derive the minimizer conditions and the uniqueness of the minimizer from only these assumptions without specifying the form of divergences if there exist the solutions that satisfy the minimizer conditions. The minimizer conditions are consistent with the results of the method of Lagrange multipliers.

First, We introduce divergence lines, balls, inner products and orthogonal subsets that are the generalization of line segments, spheres, inner products and orthogonal planes perpendicular to lines in the Euclidean space, respectively. Furthermore, we show the three-point inequality as a basic geometric property and show some properties of the divergence inner product.

Next, we discuss the minimization problem of the weighted average of divergences from some probability measures, which is important in clustering algorithms such as k-means clustering \cite{lloyd1982least}.
We show that the minimizer of the weighted average of divergences is the generalized centroid as in the case of the Euclidean space. 

Finally, we discuss the minimization problems between a probability measure $P$ and a set $\mathcal{S}$. These are interpreted as a projection from the probability measure $P$ to the set $\mathcal{S}$.
In the Euclidean space, the minimum distance from a point $P \in\mathbb {R}^3$ to a plane is given by the perpendicular foot, and the minimum distance to the sphere is given by the intersection of the sphere and a line connecting $P$ and the center of the sphere.
We show that the minimizer of divergence between a probability measure and the divergence ball or the orthogonal subset have the similar properties as in the case of the Euclidean space.

\section{Preliminaries}
This section provides definitions and notations which are used in this note.
Let $\P$ denotes the set of probability measures.
For $P,Q\in \P$, $P=Q$ denotes $P=Q \mbox{ a.s.}$.
Let $\mu$ be a dominating measure of $\P$ ($P\ll \mu$) and $p\eqdef \frac{dP}{d\mu}$ be the density of $P$. 

Divergences are defined as functions that satisfy the following properties.

Let $D: \P \times \mathcal{P}\rightarrow[0,\infty) $. 
For any $P,Q\in \mathcal{P}$,
\begin{align}
D(P\|Q)\geq 0, \nonumber \\ 
D(P\|Q)=0 \iff P=Q. \nonumber
\end{align}

\begin{definition}[Strictly convex divergence]
\label{def_strictly_convex}
Let $P, Q, R\in\P$ and $Q\neq R$.
Let $t\in(0,1)$ and let $D$ be a divergence.
The divergence $D(P\|Q)$ is strictly convex in the second argument if 
\begin{align}
(1-t)D(P\|Q) +t D(P\|R) > D(P\|(1-t)Q+t R).
\end{align} 
\end{definition}

\begin{definition}[Differentiable divergence]
Let $P, Q\in\P$ and let $D$ be a divergence.
The divergence $D(P\|Q)$ is differentiable with respect to the second argument if $D(P\|Q)$ is the functional of $q=\frac{dQ}{d\mu}$ and the functional derivative exists with respect to $q$.
Let $D[q] \eqdef D(P\|Q)$ be a functional with respect to $q$.
The functional derivative of $D(P\|Q)$ with respect to $q$ is defined by
\begin{align}
\int \frac{\delta D(P\|Q)}{\delta q(z)} \phi(z) \mathrm{d}\mu(z) \eqdef \left.\frac{d}{d\epsilon}D[q+\epsilon \phi]\right|_{\epsilon=0},
\end{align}
where $\phi$ is an arbitrary function.
\end{definition}
The strictly convex or differentiable divergence in the first argument, we can define in the same way as the second argument.
We show some examples of the functional derivative.
\begin{example}[Squared Euclidean distance]
Let $D_{\mathrm{E}}(P\|Q)\eqdef \frac{1}{2}\int (q-p)^2 \dmu$.
The functional derivative is 
\begin{align}
\frac{\delta D_{\mathrm{E}}(P\|Q)}{\delta q(z)}=q(z)-p(z).
\end{align} 
\end{example} 

\begin{example}[Bregman divergence]
\label{ex_bregman}
Let $f: \mathbb{R}\rightarrow \mathbb{R}$ be a differentiable and strictly convex function.
The Bregman divergence is defined by
\begin{align}
D_{\mathrm{B}}(P\|Q)\eqdef \int f(p) \dmu-\int f(q)\dmu -\int f'(q)(p-q) \dmu,
\end{align}
where $f'(x)$ denotes the derivative of $f$.
The Bregman divergence is strictly convex in the first argument.
The functional derivative is 
\begin{align}
\frac{\delta D_{\mathrm{B}}(P\|Q)}{\delta p(z)}=f'(p(z))-f'(q(z)).
\end{align} 
\end{example}

\begin{example}[$f$-divergence]
Let $f: \mathbb{R}\rightarrow \mathbb{R}$ be a strictly convex function and $f(1)=0$.
The $f$-divergence is defined by
\begin{align}
D_f(P\|Q)\eqdef \int qf\biggl(\frac{p}{q}\biggr) \dmu.
\end{align}
The $f$-divergence is strictly convex in the first and the second argument.
If $f$ is differentiable, the functional derivatives are
\begin{align}
\frac{\delta D_f(P\|Q)}{\delta q(z)}=\tilde{f}'\biggl(\frac{q(z)}{p(z)}\biggr),
\end{align} 
where $\tilde{f}(x)\eqdef xf\bigl(\frac{1}{x}\bigr)$ and 
\begin{align}
\frac{\delta D_f(P\|Q)}{\delta p(z)}=f'\biggl(\frac{p(z)}{q(z)}\biggr).
\end{align} 
\end{example}

\begin{example}[R{\'e}nyi-divergence]
For $0<\alpha<\infty$, the R{\'e}nyi-divergence is defined by
\begin{align}
D_\alpha(P\|Q)\eqdef \frac{1}{\alpha-1}\log\int p^\alpha q^{1-\alpha} \dmu \mbox{ for } \alpha \neq 1,\\ \nonumber
D_1(P\|Q)\eqdef\int p\log\frac{p}{q} \dmu.
\end{align}
The R{\'e}nyi divergence is strictly convex in the second argument for $0<\alpha<\infty$ (see \cite{van2014renyi}).
The functional derivative is
\begin{align}
\frac{\delta D_\alpha(P\|Q)}{\delta q(z)}=-\frac{1}{\int p^\alpha q^{1-\alpha} \dmu}{\biggl(\frac{p(z)}{q(z)}\biggr)}^\alpha.
\end{align} 
\end{example}

If a divergence $D(P\|Q)$ is differentiable or strictly convex in the first argument, by putting $\hat{D}(P\|Q)= D(Q\|P)$, $\hat{D}$ is differentiable or strictly convex in the second argument.
Hence, in the following, we only consider the differentiable or strictly convex divergences in the second argument.

\begin{definition}[Divergence line]
Let $D$ be a differentiable divergence and let $P,Q\in \P$.
The ``divergence line'' $\L(P:Q)$ is defined by 
\begin{align}
\L(P:Q)\eqdef \{R\in \P| \mbox{for } \alpha\in[0,1], \exists C(\alpha)\in\mathbb{R}, (1-\alpha) \frac{\delta D(P\|R)}{\delta r(z)} + \alpha \frac{\delta D(Q\|R)}{\delta r(z)}=C(\alpha)\}.
\end{align}
We also define probability measures on the divergence line at $\alpha$ by
\begin{align}
\L_\alpha(P,Q)\eqdef \{R\in \P| \exists C\in\mathbb{R}, (1-\alpha) \frac{\delta D(P\|R)}{\delta r(z)} + \alpha \frac{\delta D(Q\|R)}{\delta r(z)}=C\}.
\end{align}
\end{definition}
We show some examples of the divergence lines.
\begin{example}[Squared Euclidean distance]
\begin{align}
\L(P:Q)= \{R\in \P| \mbox{for } \alpha\in[0,1], R=(1-\alpha)P+\alpha Q\}.
\end{align}
\end{example}
These are mixture distributions and a line segment in Euclidean space.

\begin{example}[Kullback-Leibler divergence]
The Kullback-Leibler divergence (KL-divergence or relative entropy) \cite{kullback1951information} $D_{\mathrm{KL}}(P\|Q)$ belongs to both the Bregman divergence and the $f$-divergence.
\begin{align}
D(P\|Q)=D_{\mathrm{KL}}(P\|Q)\eqdef \int p\log\frac{p}{q} \dmu.
\end{align}
The divergence line is
\begin{align}
\L(P:Q)= \{R\in \P| \mbox{for } \alpha\in[0,1], R=(1-\alpha)P+\alpha Q\}.
\end{align}
For the reverse KL-divergence $D(P\|Q)=D_\mathrm{KL}(Q\|P)=\int q\log\frac{q}{p} \dmu$, 
\begin{align}
\L(P:Q)= \{R\in \P| \mbox{for } \alpha\in[0,1], r\eqdef \frac{dR}{d\mu}=\frac{1}{\int p^{(1-\alpha)}q^\alpha \dmu} p^{(1-\alpha)}q^\alpha \}.
\end{align}
These correspond to the m-geodesic and the e-geodesic in information geometry  
 \cite{amari2010information}.
\end{example}

\begin{definition}[Divergence inner product]
Let $D$ be a differentiable divergence.
Let $P,Q,R\in\P$.
We define ``divergence inner product'' by 
\begin{align}
\<PQ\|RQ\>\eqdef \int (q(z)-r(z))\frac{\delta D(P\|Q)}{\delta q(z)} \mathrm{d}\mu(z).
\end{align}
\end{definition}
For the squared Euclidean distance $D_{\mathrm{E}}(P\|Q)=\frac{1}{2}\int (q-p)^2 \dmu$, $\<PQ\|RQ\>=\int (p-q)(r-q) \dmu$ and this is the inner product of functions $p-q$ and $r-q$.

\begin{definition}[Orthogonal subspace]
Let $P,Q\in\P$.
We define the orthogonal subspace at $Q$ by 
\begin{align}
\O(P:Q)\eqdef \{R\in\P| \<PQ\|RQ\> =0\}.
\end{align}
\end{definition}
Since the divergence inner product is linear with respect to $R$, the orthogonal subspace is a convex set.

\begin{definition}[Divergence ball]
Let $P\in\P$ and $D$ be a divergence.
We define the divergence ball by
\begin{align}
\B_\kappa(P)\eqdef \{Q\in\P| D(P\|Q) \leq \kappa \}
\end{align}
and the surface of the divergence ball by
\begin{align}
\partial\B_\kappa(P)\eqdef \{Q\in\P| D(P\|Q)=\kappa \}.
\end{align}
\end{definition}
If the divergence is convex, the divergence ball is a convex set from the definition.

\section{Main results}
In this section, we focus on the differentiable and strictly convex divergences and show some properties of them.
We first show the three-point inequality and some properties of the divergence inner product.
Next, we discuss some minimization problems and we show that the minimizer conditions and the uniqueness of the minimizer if there exist the solutions that satisfy the minimizer conditions.

We prove the following lemma that we use in various proofs.
\begin{lemma}
\label{lem_convex}
Let $D$ be a strictly convex divergence and $P\neq Q\in\P$. Let $\lambda\in[0,1]$ and $Q_\lambda\eqdef P+(Q-P)\lambda$. Then, $D(P\|Q_\lambda)$ is strictly convex with respect to $\lambda$.
\end{lemma}

\begin{proof}
When $\lambda_1\neq \lambda_2\in [0,1]$, $Q_{\lambda_1}\neq Q_{\lambda_2}$ holds since $P\neq Q$.
From the assumption of strictly convexity of the divergence, for $t\in(0,1)$, 
\begin{align}
\label{ineq_def_convexity}
t D(P\|Q_{\lambda_1}) + (1-t)  D(P\|Q_{\lambda_2}) > D(P\|t Q_{\lambda_1} + (1-t)Q_{\lambda_2}).
\end{align}
From the definition of $Q_\lambda$, we have $t Q_{\lambda_1} + (1-t)Q_{\lambda_2}=P+(Q-P)(t\lambda_1 + (1-t)\lambda_2)=Q_{t \lambda_1+ (1-t) \lambda_2}$.
By combining this equality and (\ref{ineq_def_convexity}), the result follows.
\end{proof}

\subsection{Three-point inequality}
In this subsection, we show some geometric properties of differentiable and strictly convex divergences.

\begin{theorem}[Three-point inequality]
\label{th_three_point}
Let $D$ be a differentiable and strictly convex divergence.
Let $P,Q,R\in \P$.

Then,
\begin{align}
\label{ineq_three_point_1}
D(P\|R)\geq D(P\|Q) - \<PQ\|RQ\>,
\end{align}
where the equality holds if and only if $Q=R$.
\end{theorem}

\begin{proof}
Let $R_\lambda\eqdef Q+(R-Q)\lambda$ and $F(\lambda)\eqdef D(P\|R_\lambda)$ with $Q\neq R$.
From the assumption and Lemma \ref{lem_convex}, $F(\lambda)$ is strictly convex.
From the definition of the functional derivative for $\phi(z)=r(z)-q(z)$,  
\begin{align}
\label{eq_cos_1}
F'(\lambda)=\left.\frac{d}{d\epsilon}D(P\|R_{\lambda+\epsilon})\right|_{\epsilon=0}=\int \frac{\delta D(P\|R_\lambda)}{\delta r_\lambda(z)}(r(z)-q(z)) \dmu(z),
\end{align}
where we use $r_\lambda(z)+\epsilon(r(z)-q(z))=r_{\lambda+\epsilon}(z)$.
From the strictly convexity of $F(\lambda)$, for $\lambda >0$, we have
\begin{align}
\label{eq_cos_2}
F(\lambda)>F(0)+F'(0)(\lambda-0).
\end{align}
Substituting $\lambda=1$ into  (\ref{eq_cos_2}) and using $F(1)=D(P\|R), F(0)=D(P\|Q)$ and $F'(0)=\int \frac{\delta D(P\|Q)}{\delta q(z)}(r(z)-q(z)) \dmu=-\<PQ\|RQ\>$, we have
\begin{align}
\label{eq_cos_3}
D(P\|R)>D(P\|Q)-\<PQ\|RQ\>.
\end{align}
Hence, we have the result.
\end{proof}

For the Bregman divergence, three-point identity holds \cite{nielsen2007bregman}.
\begin{align}
D_{\mathrm{B}}(Q\|P)+D_{\mathrm{B}}(R\|Q)=D_{\mathrm{B}}(R\|P)+\int (r-q)(f'(p)-f'(q)) \dmu.
\end{align}
By using $D_{\mathrm{B}}(R\|Q)\geq 0$ and the result of Example \ref{ex_bregman}, we have the same inequality as (\ref{ineq_three_point_1}) by putting $D(P\|Q)=D_{\mathrm{B}}(Q\|P)$

We show the figure of three-point inequality.
\begin{figure}[H]
\begin{center}
\includegraphics[width=70mm, height = 50mm]{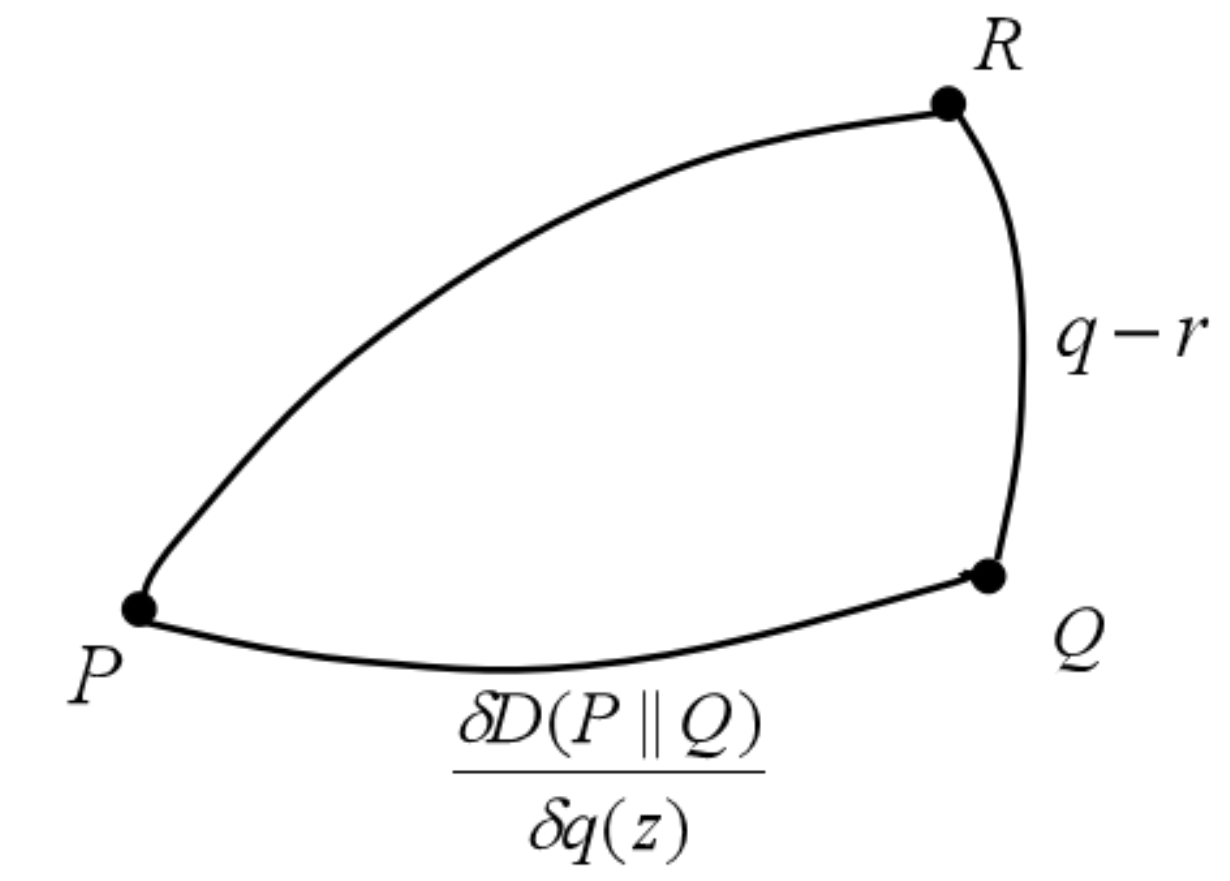} 
\caption{Three points inequality.}
\label{fig:one}
\end{center}
\end{figure}

\begin{proposition}
\label{prop_inner_product}
Let $D$ be a differentiable and strictly convex divergence.
Let $P,Q,S\in \P$ and $R\in \L_\alpha(P,Q)$ for $\alpha\in[0,1]$.

Then,
\begin{align}
(1-\alpha) \<PR\|SR\>+\alpha\<QR\|SR\>=0.
\end{align}
\end{proposition}

\begin{proof}
From the assumption, $R$ satisfies
\begin{align}
(1-\alpha) \frac{\delta D(P\|R)}{\delta r(z)} + \alpha \frac{\delta D(Q\|R)}{\delta r(z)}=C.
\end{align}
By multiplying both sides by $r(z)-s(z)$ and integrating with respect to $z$ and using $\int s\dmu=\int r\dmu=1$, the result follows. 
\end{proof}

\begin{corollary}
Let $D$ be a differentiable and strictly convex divergence and $P\neq Q\in\P$.

Then,
\begin{align}
 \<PQ\|PQ\> > D(P\|Q).
\end{align}
\end{corollary}

\begin{proof}
By substituting $P=R$ into (\ref{ineq_three_point_1}), the result follows. 
\end{proof}

\begin{proposition}
Let $D$ be a differentiable and strictly convex divergence.
Let $P\in\P$, $Q,P_\ast\in\mathcal{S}\subset \P$ and $P_\ast=\argmin_{Q\in\mathcal{S}}D(P\|Q)$.

For all $Q\in\mathcal{S}$,
\begin{align}
 \<PP_\ast\|QP_\ast\> \leq 0.
\end{align}
\end{proposition}

\begin{proof}
Let $Q_\lambda\eqdef  P_\ast+(Q-P_\ast)\lambda$ and $F(\lambda)\eqdef D(P\|Q_\lambda)$.
From the assumption, $F(0)$ is the minimum value for $\lambda\in[0,1]$.
Hence, $F'(0)\geq 0$ and we have 
\begin{align}
F'(0)=\int \frac{\delta D(P\|P_\ast)}{\delta p_\ast(z)}(q(z)-p_\ast(z))\dmu(z)=-\<PP_\ast\|QP_\ast\>\geq 0.
\end{align}
From this inequality, the result follows.
\end{proof}

This is the same approach to show the Pythagorean inequality for the KL-divergence 
 \cite{cover2012elements}.

\begin{corollary}
\label{cor_tangent_space}
Let $P, Q\in \P$ and let $R\in \L_\alpha$ for $\alpha\in(0,1)$.
Then, 
\begin{align}
\O(P:R)=\O(Q:R).
\end{align}
\end{corollary}

\begin{proof}
The result follows from Proposition \ref{prop_inner_product}.
\end{proof}

\subsection{Centroids}
We consider the minimization problem of the weighted average of differentiable and strictly convex divergences.
This problem is important for the clustering algorithm and we show that the minimizer is a generalized centroid and the uniqueness of the minimizer if there exists the solution that satisfies the generalized centroid condition.

\begin{theorem}[Minimization of the weighted average of divergences]
\label{th_centroid}
Let $D$ be a differentiable and strictly convex divergence.
Let $P_i(i=1,2,\cdots N)\in \P$ and $\alpha_i(i=1,2,\cdots N)\in \mathbb{R}$ be parameters that satisfy $\sum_i \alpha_i =1$ and $\alpha_i\geq 0$.
Suppose that there exists $P_\ast\in \P$ that satisfies 
\begin{align}
\label{eq_centroid_1}
\sum_{i=1}^N \alpha_i \frac{\delta D(P_i\|P_\ast)}{\delta p_\ast(z)}=C,
\end{align}
where $C\in \mathbb{R}$ is the Lagrange multiplier.

Then, the minimizer of the weighted average of divergences
\begin{align}
\label{eq_centroid_2}
\argmin_{R\in \P} \sum_{i=1}^N \alpha_i D(P_i\|R)=P_\ast
\end{align}
is unique and $P_\ast$ is the unique solution of (\ref{eq_centroid_1}).
\end{theorem}
\begin{proof}
We first prove (\ref{eq_centroid_2}).
Let $R\neq P_\ast$ be an arbitrary probability measure.
Let $R_\lambda\eqdef P_\ast + (R-P_\ast)\lambda$ and $F(\lambda)=\sum_{i=1}^N \alpha_i D(P_i\|R_\lambda)$.
By differentiating $F(\lambda)$ with respect to $\lambda$ and substituting $\lambda=0$, it follows that 
\begin{align}
F'(0)=\sum_{i=1}^N \alpha_i \int \frac{\delta D(P_i\|P_\ast)}{\delta p_\ast(z)}(r(z)-p_\ast(z))\dmu(z)=0,
\end{align}
where we use (\ref{eq_centroid_1}), $\int r\dmu=\int p_\ast\dmu=1$ and the definition of the functional derivative.
From Lemma \ref{lem_convex} and $\alpha_i\in[0,1]$, $F(\lambda)$ is strictly convex with respect to $\lambda$.
Hence, we have $F(1)>F(0) +F'(0)(1-0)=F(0)$.
Since $R$ is an arbitrary probability measure, from $F(1)=\sum_{i=1}^N \alpha_i D(P_i\|R)$ and  $F(0)=\sum_{i=1}^N \alpha_i D(P_i\|P_\ast)$, it follows that $P_\ast$ is the unique minimizer.
If there exists an another probability measure $\tilde{P_\ast}$ that is the solution of (\ref{eq_centroid_1}), $\tilde{P_\ast}$ is also a minimizer of (\ref{eq_centroid_2}).
This contradicts that $P_\ast$ is the unique minimizer.
Hence, the result follows. 
\end{proof}
The equality (\ref{eq_centroid_1}) is the generalized centroid condition.

\begin{proposition}
\label{prop_line}
Let $P\neq Q\in \P$ and suppose that $P_\ast\in \L_\alpha(P:Q)$ for $\alpha\in [0,1]$.
Then, $\L_\alpha(P:Q)=\{P_\ast\}$ and $\L_0(P:Q)=\{P\},\L_1(P:Q)=\{Q\}$.
\end{proposition}
\begin{proof}
By applying Theorem \ref{th_centroid} for $N=2$ and putting $P_1=P, P_2=Q$, it follows that $P_\ast$ is the unique solution of 
\begin{align}
\label{eq_line_1}
(1-\alpha)\frac{\delta D(P\|P_\ast)}{\delta p_\ast(z)}+\alpha\frac{\delta D(Q\|P_\ast)}{\delta p_\ast(z)}=C,
\end{align}
where $C\in\mathbb{R}$ is the Lagrange multiplier.
From the definition of $\L_\alpha(P:Q)$, the result $\L_\alpha(P:Q)=\{P_\ast\}$ follows. 

Next, we show that $\L_0(P:Q)=\{P\}$.
From Theorem \ref{th_centroid}, it follows that $\argmin_{R\in \P} D(P\|R)=P_\ast$.
Since $D(P\|R) \geq 0$ and $P_\ast$ is unique, we have $P_\ast=P$.
We also have the result $\L_1(P:Q)=\{Q\}$ in the same way.
\end{proof}
We can show the same theorem for the real-valued vector of $\mathbb{R}^d$.

\begin{proposition}
\label{th_centroid_vect}
Let $p_i(i=1,2,\cdots N)\in \mathbb{R}^d$ and $\alpha_i(i=1,2,\cdots N)\in\mathbb{R}$ be parameters that satisfy $\sum_i \alpha_i =1$ and $\alpha_i\geq 0$.
Let $D: \mathbb{R}^d \times\mathbb{R}^d\rightarrow[0,\infty) $ be a differentiable and strictly convex divergence.
Suppose that there exists $p_\ast\in \mathbb{R}^d$ that satisfies 
\begin{align}
\label{eq_centroid2_1}
\sum_{i=1}^N \alpha_i \frac{\partial D(p_i\|p_\ast)}{\partial p_{\ast,\nu}}=0,
\end{align}
where $\nu=\{1,2,\cdots, d\}$ and $\{p_{\ast,\nu}\}$ are components of the vector $p_\ast$. 

Then, the minimizer of the weighted average of divergences
\begin{align}
\label{eq_centroid2_2}
\argmin_{r\in \mathbb{R}^d} \sum_{i=1}^N \alpha_i D(p_i\|r)=p_\ast
\end{align}
is unique and and $p_\ast$ is the unique solution of (\ref{eq_centroid2_1}).
\end{proposition}
The proof is the same as Theorem \ref{th_centroid}.

For the Bregman divergence, $D_{\mathrm{B}}(p\|q)=f^*(q^*)-f^*(p^*)-\sum_\nu \frac{\partial f^*(p^*)}{\partial p^*_\nu}(q^*_\nu-p^*_\nu)$ holds \cite{nielsen2011burbea}, where $f:\mathbb{R}^d\rightarrow \mathbb{R}$ is a differentiable and strictly convex function, $f^*$ denotes the Legendre convex conjugate \cite{censor1997parallel} and $p^*_\nu=\frac{\partial f(p)}{\partial p_\nu}$. Since $f^*(q^*)-f^*(p^*)-\sum_\nu \frac{\partial f^*(p^*)}{\partial p^*_\nu}(q^*_\nu-p^*_\nu)=f(p)-f(q)-\sum_\nu \frac{\partial f(q)}{\partial q_\nu}(p_\nu-q_\nu)>0$ for $p\neq q$, $f^*$ is also strictly convex. Hence, $D_{\mathrm{B}}(p\|q)$ is a differentiable and strictly convex divergence with respect to $q^*$.
By combining $\frac{\partial f^*(q^*)}{\partial q^*_\nu}=q_\nu$ and (\ref{eq_centroid2_1}), it follows that 
\begin{align}
\sum_{i=1}^N \alpha_i \frac{\partial D_{\mathrm{B}}(p_i\|p_\ast)}{\partial p^\ast_{\ast, \nu}}=\sum_{i=1}^N \alpha_i (p_{\ast, \nu}-p_{i,\nu})=0.
\end{align}
Hence, $p_\ast=\sum_{i=1}^N \alpha_i p_i$ is the unique minimizer.
\subsection{Projection from a probability measure to a set}
In this subsection, we discuss the minimization problems of divergence between a probability measure and a set such as a orthogonal subset, a divergence ball or a set subject to linear moment-like constraint, we show the minimizer conditions and the uniqueness of the minimizer if there exist the solutions that satisfy the minimizer conditions.
\begin{corollary}[Minimization of the divergence between a probability measure and a orthogonal subspace]
\label{cor_orthogonal_subspace}
Let $D$ be a differentiable and strictly convex divergence and let $P,Q\in\P$.

Then, the minimizer of divergence between a probability measure $P$ and the orthogonal subspace $\O(P:Q)$
\begin{align}
\argmin_{R\in \O(P:Q)} D(P\|R)=Q
\end{align}
is unique.
\end{corollary}
\begin{proof}
Since $\<PQ\|RQ\>=0$ for $R\in \O(P:Q)$, the result follows from Theorem \ref{th_three_point}.
\end{proof}

\begin{theorem}[Minimization of the divergence between a probability measure and a divergence ball]
\label{th_divergence_ball}
Let $D$ be a differentiable and strictly convex divergence.
Let $P, Q\in \P$ and suppose that $P_\ast\in \partial\B_\kappa(P)\cap \L(P:Q)$.

Then, the minimizer of divergence between a probability measure $Q$ and a divergence ball $\B_\kappa(P)$
\begin{align}
\argmin_{R\in \B_\kappa(P)} D(Q\|R)=P_\ast
\end{align}
is unique.
\end{theorem}
\begin{proof}
Since the case $\kappa=0$ is trivial, we consider the case $\kappa>0$.
By combining Proposition \ref{prop_line} and $P_\ast\neq P$, it follows that $\alpha > 0$.
Consider an arbitrary probability measure $R\in \B_\kappa(P)$.
From the assumption and Proposition \ref{prop_line}, there exists $\alpha \in(0,1]$ and $\L_\alpha(P:Q)=\{P_\ast\}$.
Let $R_\lambda\eqdef P_\ast + (R-P_\ast)\lambda$ and $F(\lambda)=(1-\alpha)D(P\|R_\lambda)+\alpha D(Q\|R_\lambda)$.

By differentiating $F(\lambda)$ with respect to $\lambda$ and substituting $\lambda=0$, it follows that 
\begin{align}
F'(0)=\int \biggl((1-\alpha)\frac{\delta D(P\|P_\ast)}{\delta p_\ast(z)}+\alpha\frac{\delta D(Q\|P_\ast)}{\delta p_\ast(z)}\biggr)(r(z)-p_\ast(z))\dmu(z)=0,
\end{align}
where we use $\L_\alpha(P,Q)=\{P_\ast\}$, $\int r\dmu=\int p_\ast\dmu=1$ and the definition of the functional derivative.
From Lemma \ref{lem_convex} and $\alpha\in(0,1]$, $F(\lambda)$ is strictly convex with respect to $\lambda$.
Hence, we have $F(1)>F(0) +F'(0)(1-0)=F(0)$.
From $F(1)=(1-\alpha)D(P\|R)+\alpha D(Q\|R)$ and $F(0)=(1-\alpha)D(P\|P_\ast)+\alpha D(Q\|P_\ast)$, it follows that 
\begin{align}
(1-\alpha)D(P\|R)+\alpha D(Q\|R)>(1-\alpha)D(P\|P_\ast)+\alpha D(Q\|P_\ast).
\end{align}
From $D(P\|R) \leq \kappa $, $D(P\|P_\ast)=\kappa$ and $\alpha > 0$, it follows that 
\begin{align}
D(Q\|R)>D(Q\|P_\ast).
\end{align}

Since $R$ is an arbitrary probability measure in $\B_\kappa(P)$, the result follows. 
\end{proof}
%

The next corollary follows from Theorem \ref{th_divergence_ball}.
\begin{corollary}
\label{cor_tangent_balls}
Let $P, Q\in \P$ and $P_\ast\in\partial\B_{\kappa_1}(P)\cap \partial\B_{\kappa_2}(Q) \cap \L(P:Q)$.
Then, $\B_{\kappa_1}(P)\cap\B_{\kappa_2}(Q)=\{P_\ast\}$.
\end{corollary}

We show the figure that summarize Corollary \ref{cor_tangent_space}, \ref{cor_orthogonal_subspace} and  \ref{cor_tangent_balls}.

\begin{figure}[H]
\begin{center}
\includegraphics[width=70mm, height =70mm]{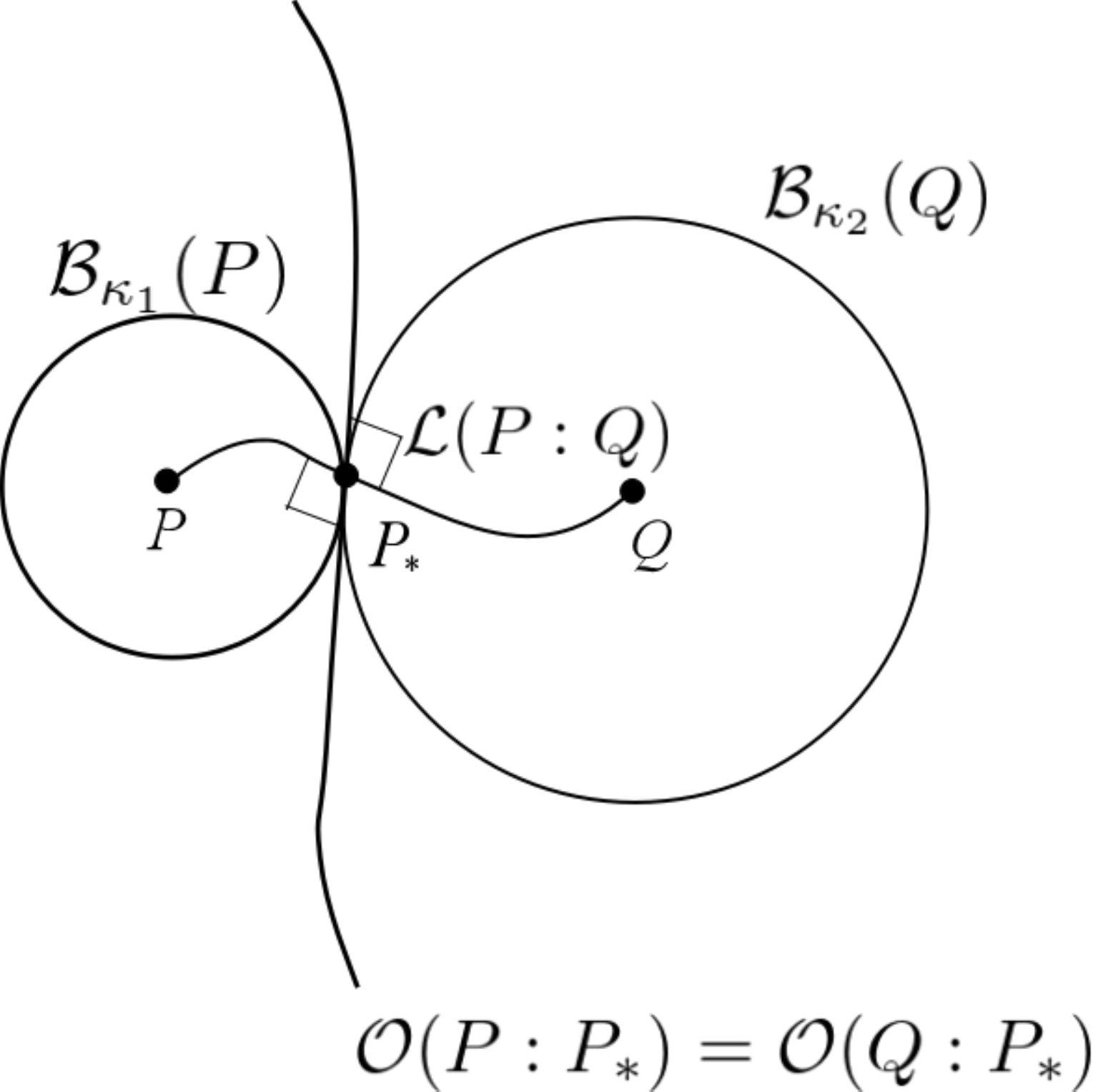}
\caption{Relation among two divergence balls, a divergence line and an orthogonal subspace.}
\label{fig:two}
\end{center}
\end{figure}

\begin{theorem}[Minimization of the divergence subject to linear moment-like constraint]
Let $D$ be a differentiable and strictly convex divergence and let $P\in\P$.
Let $T_i(Z)(i=1,2,\cdots K)$ be functions of a random variable (vector) $Z$ and $\mathcal{M}\eqdef\{Q\in \P|\mathrm{E}[T_i(Z)]=m_i (i=1,2,\cdots K), Z\sim Q\}$, where $\mathrm{E}[\cdot]$ denotes the expected value and $m_i  (i=1,2,\cdots K)\in\mathbb{R}$ are constants.

Suppose that there exists $P_\ast\in \M$ that satisfies 
\begin{align}
\label{eq_moments_1}
\frac{\delta D(P\|P_\ast)}{\delta p_\ast(z)}+\sum_{i=1}^K \beta_i T_i(z) =C,
\end{align}
where $\beta_i \in \mathbb{R}$ and $C\in\mathbb{R}$ are the Lagrange multipliers.

Then, the minimizer of divergence between a probability measure $P$ and the set $\M$ 
\begin{align}
\label{eq_moments_2}
\argmin_{R\in \M} D(P\|R)=P_\ast
\end{align}
is unique and $P_\ast$ is the unique solution of (\ref{eq_moments_1}).
\end{theorem}

\begin{proof}
We use the same technique in Theorem \ref{th_centroid}.
Consider an arbitrary $R\in \M$.
Let $R_\lambda\eqdef P_\ast+(R-P_\ast)\lambda$ and $F(\lambda)\eqdef D(P\|R_\lambda)+ \sum_i \beta_i \int T_i(z) r_\lambda(z) \dmu(z)$.

By differentiating $F(\lambda)$ with respect to $\lambda$ and substituting $\lambda=0$, it follows that 
\begin{align}
F'(0)=\int \biggl(\frac{\delta D(P\|P_\ast)}{\delta p_\ast(z)}+\sum_i \beta_i T_i(z)\biggr)(r(z)-p_\ast(z))\dmu(z)=0,
\end{align}
where we use (\ref{eq_moments_1}), $\int r\dmu=\int p_\ast\dmu=1$ and the definition of the functional derivative.
From Lemma \ref{lem_convex} and the linearity of $\sum_i \beta_i \int T_i(z) r_\lambda(z) \dmu(z)$ with respect to $\lambda$, $F(\lambda)$ is strictly convex with respect to $\lambda$.
Hence, we have $F(1)>F(0) +F'(0)(1-0)=F(0)$.
From $F(1)=D(P\|R)+ \sum_i \beta_i \int T_i(z) r(z) \dmu(z)$, $F(0)=D(P\|P_\ast)+ \sum_i \beta_i \int T_i(z) p_\ast(z) \dmu(z)$, it follows that 
\begin{align}
D(P\|R)+ \sum_i \beta_i \int T_i(z) r(z) \dmu(z)>D(P\|P_\ast)+ \sum_i \beta_i \int T_i(z) p_\ast(z) \dmu(z).
\end{align}
From $\int T_i(z) r(z) \dmu(z)=\int T_i(z) p_\ast(z) \dmu(z)=m_i$, it follows that 
\begin{align}
D(P\|R)>D(P\|P_\ast).
\end{align}

Since $R$ is an arbitrary probability measure in $\M$, it follow that $P_\ast$ is the unique minimizer.
If there exists an another probability measure $\tilde{P_\ast}$ that is the solution of (\ref{eq_moments_1}), $\tilde{P_\ast}$ is also the minimizer of (\ref{eq_moments_2}).
This contradicts that $P_\ast$ is the unique minimizer.
Hence, the result follows. 
\end{proof}

\section{Summary}
We have discussed the minimization problems and geometric properties for the differentiable and strictly convex divergences. 
We have derived the three-point inequality and introduced the divergence lines, inner products, balls and orthogonal subsets that are the generalization of lines, inner product, spheres and planes perpendicular to a line in the Euclidean space.

Furthermore, we have shown the minimizer conditions and the uniqueness of the minimizer if there exist the solutions that satisfy the minimizer conditions in the following cases, 

1) Minimization of weighted average of divergences from multiple probability measures.

2) Minimization of divergence between a probability measure and the orthogonal subsets, divergence balls, or the set subject to linear moment-like constraints.

\bibliography{reference_opt}
\end{document}